 \documentclass[a4paper]{article}
 \usepackage[margin=1.7in]{geometry}
 
\usepackage{microtype}

\usepackage{cite}
\usepackage{authblk}

\usepackage{amsmath, amssymb}
\usepackage{graphicx}
\usepackage{wrapfig}
\usepackage{comment}
\usepackage{amsthm}

\newtheorem{theorem}{Theorem}
\newtheorem{lemma}{Lemma}
\newtheorem{corollary}{Corollary}

\newcommand{\poly}{\mathrm{poly}}

\begin{document}

\title{On the expressive power  of quasiperiodic  SFT\footnote{Supported by ANR-15-CE40-0016-01 RaCAF grant}}

\author[1]{Bruno Durand}
\author[2]{Andrei Romashchenko}
\affil[1]{Univ. Montpellier \& LIRMM}
\affil[2]{CNRS \& LIRMM, on leave from IITP RAS}

\maketitle

\begin{abstract}
In this paper we study the \emph{shifts}, 
which are the shift-invariant and topologically closed sets of configurations over a finite alphabet in  $\mathbb{Z}^d$. 
The \emph{minimal shifts} are those shifts in which all configurations contain exactly  the same patterns. 
Two  classes of  shifts play a prominent role in symbolic dynamics, in language theory  and in the theory of computability:
the \emph{shifts of finite type} (obtained by forbidding a finite number of finite patterns) and the \emph{effective shifts}
(obtained by forbidding a computably enumerable set of finite patterns). 
We prove that every effective minimal shift can be represented as a factor of a projective subdynamics on a minimal shift \emph{of finite type} in a bigger (by $1$) dimension.
This result transfers to the class of minimal shifts a theorem by M.~Hochman known for the class of all effective shifts and
thus answers an open question by E.~Jeandel. 
We  prove a similar result for quasiperiodic shifts and also show that there exists a quasiperiodic shift of finite type for which  Kolmogorov complexity  of 
 all patterns of size $n\times n$ is $\Omega(n)$. 
\end{abstract}

\section{Introduction}

The  study of symbolic dynamics  was initially motivated as a discretization of classic dynamical
systems, \cite{morse-hedlund}.  Later,  the focus of attention in this area shifted towards the questions 
related to computability theory. 
The central notion of symbolic dynamics is a \emph{shift} (a.k.a. \emph{subshift}), which is
a set of configurations in $\mathbb{Z}^d$ over a finite alphabet, defined by a set of forbidden patterns.
Two major notions -- two classes of shifts -- play now a crucial role 
in symbolic dynamics: shifts of \emph{finite type} (SFT,  the shift defined by a finite set of forbidden patterns)
and \emph{effective} shifts (a.k.a. \emph{effectively closed}  --- shifts with an enumerable set of forbidden patterns). 
These classes are distinct:
every SFT is 
effective, but in general the reverse implication does not hold. 
However,  the differences between these classes is surprisingly subtle. It is known
that every effective shift can be {simulated} in some sense by an SFT of higher dimension.
More precisely, every effective shift  in $\mathbb{Z}^d$ can be represented as a factor of the projective subdynamics  of 
an SFT of dimension increased by $1$, see \cite{hochman-2,drs,aubrun-sablik}.

Usually, the proofs of computability results in symbolic dynamics involve  sophisticated algorithmic gadgets embedded in dynamical systems.
The resulting constructions are typically intricate and somewhat artificial. 
So, even if the shifts (effective 
or SFT)  in general 
are proven to have a certain algorithmic property, the known proof  may be inappropriate for   ``natural''  dynamical systems. 
Thus, it is interesting to understand the limits of the known algorithmic techniques and find out whether the remarkable
properties of algorithmic complexity can be extended to ``simple'' and mathematically ``natural'' types of shifts.

One of the classic natural types of dynamical systems is the class of \emph{minimal shifts}. Minimal shifts are those containing no proper
shift, or equivalently the shifts  where all configurations have exactly the same patterns. The role minimal shifts play in symbolic dynamics is  
similar to the role simple groups play in group theory (in particular  every  nonempty shift contains a nonempty
minimal shift, see a discussion in \cite{bruno}). Notice that all minimal shifts are quasiperiodic (but the converse is not true).
Intuitively it seems that the structure of a  minimal shift must be  simple (in terms of dynamical systems).

Besides, minimal shifts cannot be ``too complex'' in algorithmic terms. 
Indeed,  it is known that every effective minimal shift has a computable language of patterns, 
and it contains at least one computable configuration \cite{hochman-2}
(which is in general not the case for effective shifts and even for SFT).
Nevertheless,  minimal shifts can have quite nontrivial  algorithmic properties \cite{jeandel-vanier, hochman-vanier}.

We have mentioned above that every effective shift $\cal S$  can be represented as a factor of a projective subdynamics  of 
an SFT ${\cal S}'$ (of higher dimension). In the previously known proofs of this result  \cite{hochman-2,drs,aubrun-sablik}, 
even if  $\cal S$  is  minimal, 
the structure of the corresponding  SFT ${\cal S}'$   (that simulates by its projective subdynamics  the given $\cal S$)  
can be very sophisticated (and far from being  minimal). 
So, a natural question arises (E.~Jeandel, \cite{jeandel-email}): 
is it true that every effective minimal (or quasiperiodic) shift can be represented as a factor of a projective subdynamics on a
minimal (respectively, quasiperiodic) SFT of higher dimension?  In this paper we give a positive answer to that questions.

The full proof of the main result of this paper is rather cumbersome  for the following reason:
we use the  technique of self-simulating tilings
(e.g., \cite{drs, mfcs2015, westrick}) combined  with some combinatorial lemmas on
quasiperiodic configurations.  Unfortunately, there is no simple and 
clean separation between 
the generic technique of self-simulating tilings and
the supplementary features embedded in this type of tilings,  
so  we cannot use the (previously known) technique of self-simulation as a ``black box''.
We have to re-explain  the core techniques of  fixed-point programming embedded in tilings and adjust the  supplementary features within the construction.

While explaining the proofs, we have to balance clarity with formality, and given the usual space limits of 
the conference paper\footnote{A short version of this paper is to be presented at the MFCS~2017.} 
we have to sketch some 
standard parts of the proof.
In Appendix, we provide a somewhat more verbose explanation of the basic construction of the self-simulating tilings.

\subsection{Notation and basic definitions}

Let $\Sigma$ be a finite set (an alphabet). Fix an integer $d>0$. A $\Sigma$-\emph{configuration} is a mapping
$
 \mathbf{f} \ :\ \mathbb{Z}^d \to \Sigma.
$
(i.e.,  a coloring of $\mathbb{Z}^d$ by ``colors'' from $\Sigma$). 
The set of all $\Sigma$-\emph{configurations} is called \emph{the full shift}. 

A $\mathbb{Z}^d$-\emph{shift}  (or just a \emph{shift} if $d$ is clear from the context) is a set of configuration that is 
 (i)~shift-invariant (with respect to the translations along each  coordinate axis), and
 (ii)~closed in Cantor's topology.

 A \emph{pattern} is a mapping from a finite subset in $\mathbb{Z}^d$ to $\Sigma$ (a coloring of a finite set of $\mathbb{Z}^d$). 
 Every shift can be  defined by a set of forbidden finite patterns $F$ (a configuration belongs to the shift if and only if it does
not contain any pattern from $F$).
A shift is called  \emph{effective} (or \emph{effectively closed}) if it can be defined by a computably enumerable  set of forbidden patterns. A shift is called a \emph{shift of finite type} (SFT), if it can be defined by a finite set of forbidden patterns. 

A special class of a $2$-dimensional SFT is defined in terms of \emph{Wang tiles}. In this case we interpret the alphabet $\Sigma$ as a set of \emph{tiles} --- unite squares with colored sides, assuming that all colors belong to some finite set $C$ (we assign one color to each side of a tile, so technically $\Sigma$ is a subset of $C^4$). A (valid) \emph{tiling} is  a set of all configurations
 $
 \mathbf{f}\ :\ \mathbb{Z}^2\to \Sigma
 $
where every two neighboring tiles match, i.e., share the same color on the adjacent sides. 
Wang tiles are powerful enough to simulate any SFT in a  very strong sense: for each SFT $\cal S$ there exists a set of Wang tiles $\tau$ such that the 
set of all $\tau$-tilings is isomorphic to $\cal S$.  In this paper we focus on  tilings since Wang tiles perfectly suit the technique of self-simulation.

A shift $\cal S$ (in the full shift  $\Sigma^{\mathbb{Z}^d}$) can be interpreted as a dynamical system. There are $d$  shifts along each of the coordinates, and each of these shifts  map  $\cal S$ to itself. So, the group $\mathbb{Z}^d$ naturally acts on $\cal S$. 

For any shift $\cal S$  on $\mathbb{Z}^d$ and for any  $k$-dimensional sublattice  $L$ in $\mathbb{Z}^d$, 
the $L$-\emph{projective subdynamics}  ${\cal S}_L$ of $\cal S$ is the set of configurations of $\cal S$ restricted on $L$.
The $L$-projective subdynamics of a $\mathbb{Z}^d$-shift can be understood as a $\mathbb{Z}^k$-shift
(notice that $L$ naturally acts on ${\cal S}_L$).
 In particular, for every $d'<d$ we have a \emph{standard} $\mathbb{Z}^{d'}$-projective subdynamics on the shift $\cal S$ generated by the lattice spanned on the first $d'$ coordinate axis.
In the proofs of Theorems~\ref{thm-main}-\ref{thm-main-min} we deal with the standard $\mathbb{Z}^{(d-1)}$-projective subdynamics on $\mathbb{Z}^d$-shifts.
In this paper we focus mostly on $2$-dimensional shifts and on $\mathbb{Z}^{1}$-projective subdynamics on these shifts.

A configuration $\omega$ is called \emph{recurrent} if every pattern that appears in $\omega$ at least once, must then appear in this configuration
infinitely often. A configuration $\omega$ is called \emph{quasiperiodic} (or \emph{uniformly recurrent}) if every pattern $P$ that appears in $\omega$ at least once, must appear in every pattern $Q$ large enough in $\omega$. Notice that every periodic configuration is also quasiperiodic.
It is easy to see that if a shift $\cal S$ is minimal, then every $\omega \in {\cal S}$ is  quasiperiodic. The converse, in general, is not true.

For a quasiperiodic  configuration $\omega$, its \emph{function of a quasiperiodicity} is a mapping
 $
  \varphi \ : \ \mathbb{N} \to \mathbb{N}
 $  
such that  every finite pattern of diameter $n$  either never appears in  $\omega$, or it appears in every
pattern of size  $\varphi(n)$ in $\omega$, see \cite{bruno}.
(A function of quasiperiodicity is not unique: increasing a function of quasiperiodicity of $\omega$ at each point we get again a function of quasiperiodicity of this configuration.)
Similarly, a shift $\cal S$ has a function of quasiperiodicity $ \varphi$, if $ \varphi$ is a function of a quasiperiodicity
for every configuration  in ${\cal S}$. 

If a shift $\cal S $ is minimal, then all configurations in $\cal S$ have exactly the same finite patterns.
For  every minimal  shift $\cal S$,  the function of quasiperiodicity is finite (for every $n$) and even
computable. Moreover,  for an effective  minimal  shift, the set of all finite patterns (that can appear in any configuration)  is computable, see \cite{hochman-2, ballier-jeandel}. 
From this fact it follows that every effective and minimal  shift  contains a computable  configuration.
Indeed, with an algorithm that checks whether patterns appear in every $\omega\in{\cal S}$, we can  incrementally (and algorithmically) increase a finite pattern, maintaining the property  that this pattern appears in every configuration in $\cal S$.

If a non-minimal effective  shift contains only quasiperiodic configurations, then we can claim that its function of quasiperiodicity is finite and even computable. However,  we cannot guarantee that the set of all finite patterns (that  appear in at least one configuration)  is computable,  see \cite{ballier-jeandel}.

\subsection{The main results}

Our first theorem claims that every effective quasiperiodic $\mathbb{Z}^d$-shift   can be simulated by a quasiperiodic SFT in $\mathbb{Z}^{d+1}$.
\begin{theorem}\label{thm-main}
Let $\cal A$ be an effective quasiperiodic  $\mathbb{Z}^d$-shift  over some alphabet $\Sigma_A$. Then  there exists a quasiperiodic SFT $\cal B$ 
 \textup(over another alphabet $\Sigma_B$\textup) of dimension $d+1$ such that
$\cal A$ is isomorphic to a factor of a $d$-dimensional projective subdynamics on $\cal B$. 
\end{theorem}
\noindent
A similar result holds for  effective minimal shifts:
\begin{theorem}
\label{thm-main-min}
For every effective minimal  $\mathbb{Z}^d$-shift $\cal A$  there exists a minimal SFT  $\cal B$ 
 in $\mathbb{Z}^{d+1}$ such that
$\cal A$ is isomorphic to a factor of a $d$-dimensional  projective subdynamics on  $\cal B$. 
\end{theorem}

In the proof of Theorems~\ref{thm-main}-\ref{thm-main-min} 
we deal with the projective subdynamics of $\cal B$ corresponding to the lattice on the first $d$ coordinate axis.
So, more technically,  Theorem~\ref{thm-main} claims that 
there exists a projection $\pi \ : \ \Sigma_B\to \Sigma_A$ such that for every configuration
 $
  \textbf{f}\ : \ \mathbb{Z}^{d+1} \to \Sigma_B
 $
from $\cal B$ and for all $i_1,\ldots,i_d,j,j'$  
we have $\pi (\textbf{f}(i_1,\ldots,i_d,j)) = \pi(\textbf{f}(i_1,\ldots,i_d,j'))$ 
(i.e., the projection $\pi$ takes a constant value  along   each column  $(i_1,\ldots,i_d,*)$), and the resulting $d$-dimensional configuration 
$
 \{ \pi(\textbf{f}(i_1,\ldots,i_d,*)) \}
$
belongs to $\cal A$; moreover, each configuration of $\cal A$ can be represented in this way by some configuration of $\cal B$.  
Informally, we say that each  configuration from $\cal B$ ``encodes'' a configuration from $\cal A$, and each configuration from $\cal A$ is encoded by some configuration from $\cal B$. 

\smallskip

Theorem~\ref{thm-main} implies the following  somewhat surprising corollary 
(a quasiperiodic  $\mathbb{Z}^2$-SFT can have highly ``complex'' languages of patterns):
\begin{corollary}\label{thm-kolmogorov}
There exists a quasiperiodic SFT $\cal A$ of dimension $2$ such that
Kolmogorov complexity of every $(N\times N)$-pattern in every configuration of  $\cal A$ is   $\Omega(N)$.
\end{corollary}
In other words, a quasiperiodic  $\mathbb{Z}^2$-SFT can have highly ``complex'' languages of patterns.

\emph{Remark 1:}  A standalone pattern of size $N\times N$ over an alphabet $\Sigma$  (with at least two letters) can have a Kolmogorov
complexity up to $\Theta(N^2)$. However, this density of information cannot be enforced by  local rules, because in every SFT in $\mathbb{Z}^2$ there exists
a  configuration such that Kolmogorov complexity of all $N\times N$-patterns is bounded by $O(N)$, \cite{dls}. Thus, the  lower bound $\Omega(N)$
 in Corollary~\ref{thm-kolmogorov} is optimal in the class of all SFT.

\emph{Remark 2:} Every effective (effectively-closed) minimal  shift $\cal A$ is \emph{computable}: given a pattern, we can algorithmically decide whether it belongs to the configurations of the shift 
(which is in general not the case for effective quasiperiodic shift). 
Patterns of high Kolmogorov complexity cannot be found algorithmically. So Corollary~\ref{thm-kolmogorov} cannot be extended to the class of minimal SFT.

\smallskip

To simplify notation and make the argument more visual, in what follows we focus on the case $d=1$. The proofs extend to  any $d>1$ in a straightforward way, \emph{mutatis mutandis}.

\smallskip

The rest of the paper is organized as follows. In Section~2 we briefly remind the core technique of the self-simulating tilings from \cite{drs}
(see also a more detailed explanation in Appendix). In Section~3 we explain how to ``simulate'' any given effective $\mathbb{Z}$-shift in
a $\mathbb{Z}^2$-SFT (the simulated one-dimensional shift is obtained as a factor of the standard  projective subdynamics on the constructed
$\mathbb{Z}^2$-SFT);  the technique discussed in this section also previously appeared in \cite{drs}. In Section~4 we formulate and prove two combinatorial lemmas concerning
quasiperiodic sequences (this part of the argument can be understood independently of the previous sections). At last, in Section~5 we
combine the proven combinatorial lemmas with the technique from \cite{mfcs2015} (enforcing quasiperiodicity or minimality in a self-simulating tiling).
This combination results in a proof of our main results, Theorem~\ref{thm-main} and Theorem~\ref{thm-main-min}.

\section{The general framework of self-simulating SFT}

In what follows we extensively use  the technique of \emph{self-simulating tilesets} from \cite{drs} (this technique  goes back to \cite{gacs-self-similar-automata}). We use the idea of {self-simulation} to enforce  a kind of   \emph{self-similar} structure in a tiling.  In this section we remind the reader of the principal ingredients of this construction.

Let $\tau$ be a tileset and $N>1$ be an integer.  We call a $\tau$-\emph{macro-tile} an $N \times N $ square  correctly tiled by tiles from $\tau$. 
Every side of a $\tau$-macro-tile contains  a sequence of $N$ colors (of tiles from $\tau$); we refer to this sequence as a \emph{macro-color}.
A tileset $\tau$ \emph{simulates} another tileset $\rho$, if there exists a set of $\tau$-{macro-tiles} $T$ such that
 \begin{itemize}
\item there is one-to-one correspondence between $\rho$ and $T$ (the colors of two tiles from $\rho$ match if and only if the macro-colors 
 of the corresponding macro-tiles from $T$ match), 
 \item  for every $\tau$-tiling there exists a unique lattice of vertical and horizontal lines that splits this tiling into $N\times N$ macro-tiles from $T$,
 i.e., every $\tau$-tiling represents a unique $\rho$-tiling.
 \end{itemize}
For a large class of sufficiently ``well-behaved'' sequence of integers $N_k$ we can construct a family of tilesets $\tau_k$ ($i=0,1,\ldots$)  such that each $\tau_{k-1}$ simulates the next $\tau_{k}$ with the zoom $N_k$ (and, therefore, $\tau_0$ simulates every $\tau_k$ with the zoom $L_k = N_1\cdot N_2 \cdots N_{k}$).

If a $k$-level macro-tile $M$ is a ``cell'' in a $(k+1)$-level macro-tile $M'$, we refer to  $M'$ as a \emph{father} of $M$; we call the $(k+1)$-level macro-tiles neighboring $M'$ \emph{uncles} of $M$.

 In our construction each tile of $\tau_k$  ``knows'' its coordinates modulo $N_k$ in the tiling:  the colors on the left and on the bottom sides should involve $(i,j)$, the color on the right side should involve $(i+1\mod N_k, j)$, and the color on the top side, respectively, involves $(i, j+1 \mod N_k)$. 
So every $\tau_k$-tiling can be uniquely split into blocks (macro-tiles) of size $N_k\times N_k$, where the coordinates of cells range from $(0,0)$ in the bottom-left corner to $(N-1,N-1)$ in top-right corner. Intuitively, each tile ``knows'' its position in the corresponding macro-tile.

\begin{wrapfigure}{l}{0.48\textwidth}
 \vspace{-5pt}
\includegraphics[scale=0.48]{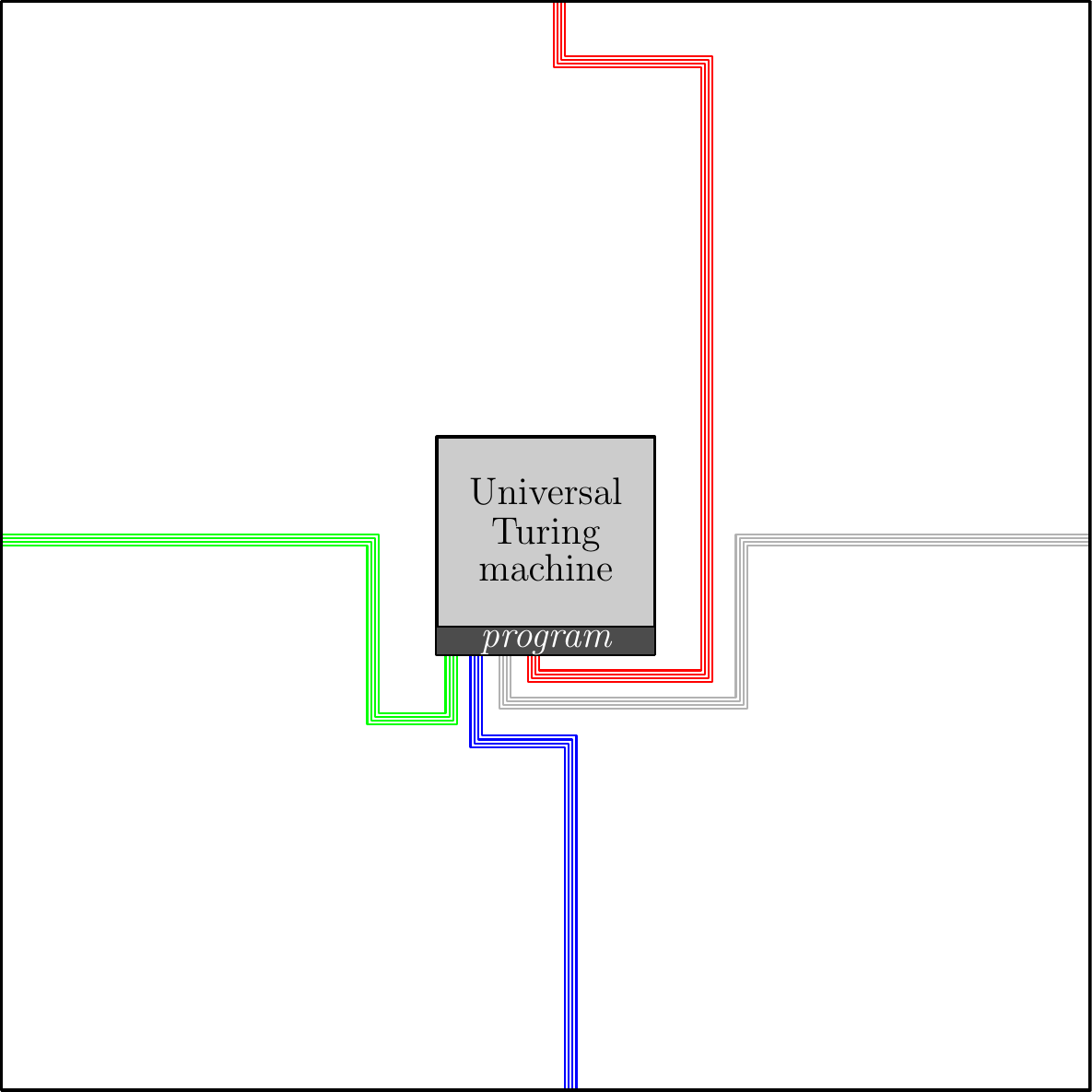}
\label{pic4}
 \vspace{-7pt}
\caption{The structure of a macro-tile.}
\vspace{-5pt}
 \end{wrapfigure}

In addition to the coordinates, each tile in $\tau_k$  has some supplementary information   encoded in the colors on its sides (the size of the supplementary information is always bounded by $O(1)$). 
In the middle of each side of a macro-tile we allocate $s_k \ll N_k$ positions  where an array of $s_k$ bits 
represents a color of a tile from $\tau_{k+1}$
(these $s_k$ bits are embedded  in colors on the sides of $s_k$ tiles  of a macro-tile, one bit per a cell). 
We fix some cells in a macro-tile that serve as  ``communication wires'' 
and then require that these tiles carry the same (transferred) bit on two sides (so the bits of ``macro-colors'' are transferred from the sides of macro-color towards its central part). The central part of a macro-tile (of size, say $m_k\times m_k$, where $m_k = \poly(\log N_k)$)  is a \emph{computation zone}; it represents a space-time diagram of a universal Turing machine (the tape is horizontal, time goes up), see Fig.~1.

The first line of the computation zone contains the following \emph{fields} of the \emph{input data}:
 \begin{itemize}
 \item[(i)] the program of a Turing machine $\pi$ that verifies that a quadruple of macro-colors correspond to one valid macro-color,
 \item[(ii)] the binary expansion of the integer rank $k$ of this macro-tile,
 \item[(iii)] the bits encoding the macro-colors ---  the position inside the ``father'' macro-tile of rank $(k+1)$ (two coordinates modulo $N_{k+1}$) and 
 $O(1)$ bits of the supplementary information assigned to  the macro-colors.
 \end{itemize}
We require that the simulated computation terminates in an accepting state (if not, no correct tiling can be formed).
The simulated computation guarantees that macro-tiles of level $k$ are isomorphic to the tiles of $\tau_{k+1}$. Notice that on each level
$k$ of the hierarchy we simulate in macro-tiles a computation of one and the same Turing machine $\pi$. Only the inputs for this machine
(including the binary expansion of the rank number $k$) varies on different levels of the hierarchy. 

This construction of a tileset can be implemented
using the standard technique of self-referential programming, similar to the Kleene recursion theorem, as it is shown in \cite{drs}.
The construction works if the size of a macro-tile (the  zoom factor $N_k$) is large enough. First, we need enough space in a macro-tile to ``communicate''  $s_k$ bits  from each macro-colors to the computation zone; second, we need  a large enough computation zone, so all accepting computations  terminate in time $m_k$ and on space $m_k$. In what follows we assume that $N_k = 3^{C^k}$ for some large enough $k$. 

\section{Embedding a bi-infinite sequence into a self-simulating tiling}

In this section we adapt the technique from  \cite{drs} and explain how to ``encode'' in a self-simulating tiling a bi-infinite sequence, 
and provide to the computation zones of macro-tiles of all ranks an access to the letters of the embedded sequences.

We are going to embed in our tiling a bi-infinite sequence $\mathbf{x} = (x_i)$ over an alphabet $\Sigma$.
To this end we assume that each $\tau$-tile ``keeps'' a letter from $\Sigma$ that propagates without change in the
vertical direction.  Formally speaking, a letter from $\Sigma$ should be a part of the top and bottom colors of every $\tau$-tile 
(the letters assigned to both sides of a tile must be equal to each other).  We want to guarantee that a $\Sigma$-sequence
 can be embedded in a $\tau$-tiling, if and only if it belongs to some fixed effective $\cal A$. 
 (We postpone to Section~5 the discussion of \emph{quasiperiodicity} of embedded shifts.)

  \begin{wrapfigure}{l}{0.55\textwidth}
  \vspace{0pt}
\includegraphics[scale=0.49]{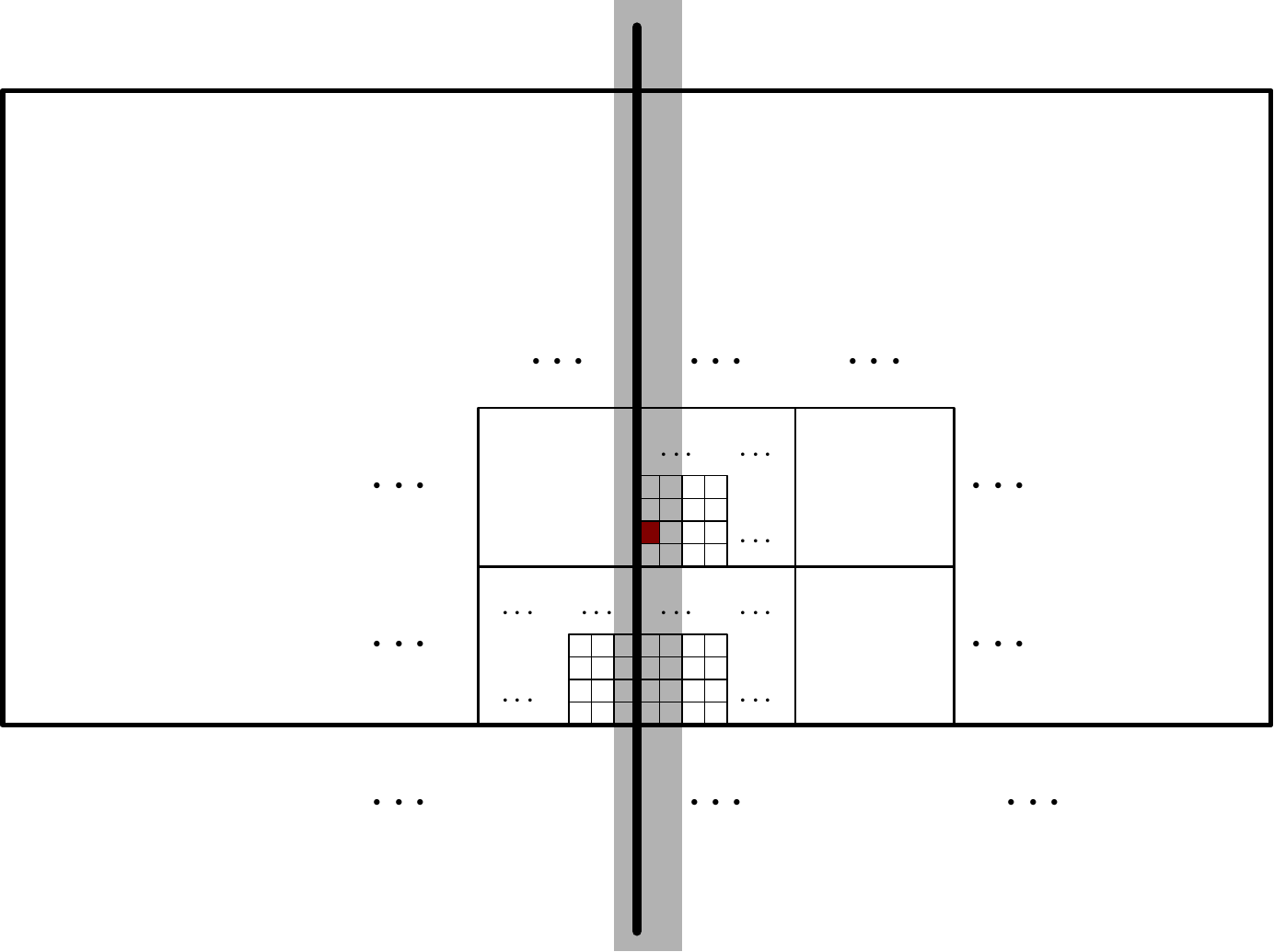}
\caption{The zone of responsibility (the grey vertical stripe)  for a macro-tile (the red square) is $3$ times wider than the macro-tile itself.}
\vspace{-8pt}
\end{wrapfigure}

We want to ``delegate'' the factors of the embedded sequence to the computation zones of macro-tiles,
where these factors will be validated (that is, we will check that they do not contain any forbidden subwords). 
While using tilings with growing zoom factor, we can guarantee that  the size of the computation zone of a $k$-rank macro-tile grows with the rank $k$. 
So we have at our disposal the computational resources suitable to run all necessary validation tests on the embedded sequence.
It remains to organize the propagation of the letters of the embedded sequence to the ``conscious memory'' 
(the computation zones) of macro-tiles of all ranks. In what follows we explain how this propagation is organized.

 \emph{Zone of responsibility  of macro-tiles.}\label{letter-delegation}
In our construction,  a macro-tile of level $k$ is a square of size $L_k\times L_k$, with 
$L_k=N_1\cdot N_2\cdot \ldots\cdot N_{k}$ (where $N_i$ is the zoom factor on level $i$ of the hierarchy of macro-tiles).
We say that a $k$-level macro-tile is \emph{responsible} for the letters of the embedded sequence $\mathbf{x}$
assigned to the columns of (ground level) tiles of this macro-tile as well as  to the columns of macro-tiles of the same rank 
on its left and on its right.
That is, the \emph{zone of responsibility} of a $k$-level macro-tile is a factor of length $3L_k$
from the embedded sequence, see Fig.~2. (The zones of responsibility of any two horizontally neighboring macro-tiles overlap.)

\smallskip
\emph{Letters assignment:}\label{letter-assignment}
The computation zone of a $k$-level macro-tile (of size $m_k\times m_k$) is too small to contain all letters from its zone 
of responsibility. So we require that the computation zone obtains as an input a (short enough) chunk of letters from its zone 
of responsibility. Let us say, that it is a factor of length $l_k = \log \log L_k$ from the stripe of $3L_k$ columns constituting the zone
of responsibility of this macro-tile. We say that this chunk is \emph{assigned} to this macro-tile.

The infinite stripe of vertically aligned $k$-level macro-tiles share the same zone of responsibility.  However, different macro-tiles 
in such a stripe will obtain different assigned chunks. The choice of the assigned chunk varies from $0$ to $(3L_{k}-l_k)$. We need
to choose a position of a factor of length $l_k$ in a word of length $L_k$.  Let us say  for certainty that for a macro-tile $M$ of rank $k$
the first position of the assigned chunk (in the stripe of length $3L_k$) is defined as the vertical position of $M$ in the bigger macro-tile
of  rank $(k+1)$ (modulo $(3L_{k}-l_k)$). 

\emph{Remark:} We have chosen the zoom factors $N_k$ so that $N_{k+1}\gg 3L_k$. Hence, every chunk of length $l_k$ from
a stripe of width $3L_k$ is assigned to some of the macro-tiles ``responsible''  for these $3L_k$ letters. 
Since the zones of responsibility of neighboring $k$-level macro-tiles overlap by more than $l_k$, every finite factor of length $l_k$ in the embedded sequence $\mathbf{x}$  is assigned  to some $k$-level macro-tile (even if it involves columns of two macro-tiles of rank $k$).

\smallskip

\emph{Implementing the letters assignment by self-simulation.} In the \emph{letters assignment} paragraph above we presented some requirements --- how the data must be propagated from the ground level (individual tiles) to $k$-level macro-tiles.  Technically, for each $k$-level macro-tile $\cal M$ we specified which chunk of the embedded sequence should be a part of the data fields on the computation zone of  $\cal M$. So far we have not explained how the assigned chunks arrive  to the high-level data fields. 
 Now, we are going to explain how to implement the desired scheme of letter assignment in a self-simulating tiling. 
 Technically, we append to the input data of the computation zones of macro-tiles some supplementary data fields:
 \begin{itemize}
 \item[(iv)] the block of $l_k$ letters from the embedded sequence assigned to this macro-tile,
 \item[(v)] three blocks of bits of $l_{k+1}$ letters  of the embedded sequence assigned to this ``father'' macro-tile,
 and two ``uncle''  macro-tiles (the left and the right neighbors of the ``father''),
 \item[(vi)] the coordinates of the ``father'' macro-tile in the ``grandfather'' (of rank $(k+2)$). 
 \end{itemize}
Informally, each $k$-level macro-tile must check that the data in the fields (iv), (v) and (vi) is consistent. That is, if some letters from the fields 
(iv) and (v) correspond to the same vertical column (in the zone of responsibility), then these letters must be equal to each other. Also,
if a $k$-level macro-tile plays the role of cell in the computation zone of the $(k+1)$-level father, it should check the consistency of
its (v) and (vi) with the bits displayed in  father's computation zone.  Finally, we must ensure the coherence of  the fields (v) and (vi)
for each pair of neighboring $k$-level macro-tiles; so this data should make a part of the macro-colors.

Notice that the data from ``uncles'' macro-tiles is necessary to deal with the letters from the columns that physically  belong to the
neighboring macro-tiles. So the consistency of the fields (v) is imposed also on  neighboring  $k$-level macro-tiles that belong to 
different  $(k+1)$-level fathers (the boarder line between these  $k$-level macro-tiles is also the boarder line between their fathers).

The  computations  verifying the coherence of the new fields can be performed in polynomial time, and 
the required update of the construction fits the constraints on the parameter. 
See a more detailed discussion on ``letter delegation''  in \cite[Section~7]{drs}.

\smallskip

\emph{Final remarks: testing against forbidden factors.}
To guarantee that the embedded sequence $\mathbf{x}$ contains no forbidden patterns, each $k$-level 
macro-tile should allocate some  part of its computation zone to enumerate (within the limits of available space and time) 
the forbidden pattern, and verify that the block of $l_k$ letters assigned to this macro-tile contains none of the  found forbidden factors.

The time and space allocated to  enumerating the  forbidden words  grow as a function of $k$. To ensure that the embedded sequence
contains no forbidden patterns,  it is
enough to guarantee that each forbidden pattern is found by macro-tiles of high enough rank, and every factor of the embedded sequence
is compared (on some level of the hierarchy) with every forbidden factor.
Thus, we have a general construction of a $2$D tiling that simulates a given, effective $1$D shift. 
In the next sections we explain how to make these tilings quasiperiodic in the case when the simulated $1$D shift is also quasiperiodic.

\section{Combinatorial lemmas: the direct product of quasiperiodic and periodic sequences}
The technique  from \cite{drs} allows to embed in a self-similar tiling a $1$-dimensional sequence and handle factors of this sequence. However, the previously known constructions cannot guarantee minimality or quasiperiodicity of the resulting tiling, even if the embedded sequences have very simple combinatorial structure. To achieve the property of  quasiperiodicity we will need some new techniques. The new parts of the argument begins with two  simple combinatorial lemmas concerning quasiperiodic sequences.  
 These lemmas  are isolated results, which can be formulated and proven independently of our main construction. 
\begin{lemma}
\label{lemma-quasiperiodic-times-periodic}
 \textup(see \cite{avgust,salimov}\textup)
 Let $\mathbf{x}$ be a bi-infinite  recurrent sequence,  $v$ be a finite factor  
in $\mathbf{x}$, and $q$ be a positive integer number.
Then there exists an integer $t>0$ such that another copy of $v$ appears in  $\mathbf{x}$ with a shift $q\cdot t$. In other words, there exists
another instance of the same factor $w$ with a shift divisible by $q$.
Moreover, if $\mathbf{x}$ is quasiperiodic, then the gap $q\cdot t$ between neighboring appearances of $v$ is bounded by some number $L$ that depends on $\mathbf{x}$ and $v$ \textup(but not on a specific instance of the factor $x$ in the sequence\textup).
\end{lemma}

\noindent
\emph{Notation:}
For a configuration $\mathbf{x}$ (over some finite alphabet) we denote with ${\cal S}(\mathbf{x})$ the shift that consists of all configurations $\mathbf{x}'$ containing only patterns from $\mathbf{x}$. If a shift $\cal T$ is minimal, then ${\cal S}(\mathbf{x}) = {\cal T}$ for all configurations $\mathbf{x}\in {\cal T}$.

\begin{lemma}
\label{lemma-minimal-times-periodic}
(a) Let  $\cal T$ be an effective minimal shift. Then for every  $\mathbf{x}=(x_i)$ from ${\cal T}$ and  every periodic configuration $\mathbf{y}=(y_i)$
the direct product $\mathbf{x} \otimes \mathbf{y}$  \textup(the bi-infinite sequence of pairs $(x_i,y_i)$ for $i\in \mathbb{Z}$\textup) generates a minimal
shift, i.e., ${\cal S}(\mathbf{x} \otimes \mathbf{y}) $  is minimal.
(b) If in addition 
the sequence $\mathbf{x}$ are computable, then  the set of patterns in ${\cal S}(\mathbf{x} \otimes \mathbf{y}) $  is also computable.
\end{lemma}

\noindent
\emph{Remark:} In general, different configurations $\mathbf{x} \in {\cal T}$ in the product with one and  the same periodic $\mathbf{y}$ can result in different
shifts ${\cal S}(\mathbf{x} \otimes \mathbf{y}) $.

\smallskip


\begin{proof}[Proof of lemma~\ref{lemma-quasiperiodic-times-periodic}]
Denote by  $N$ the length of $v$.
We are given that $v$ is a factor of $\mathbf{x}$. W.l.o.g., we may assume that
 $
 v = \mathbf{x}_{[0: N-1]}.
 $
We need to prove that $v$ reappears again in $\mathbf{x}$ with a shift $t\cdot q$, i.e.,  
$
v = \mathbf{x}_{[t q : t q + N-1]}
$
for some $t>0$.

 Since $\mathbf{x}$ is quasiperiodic, there exists an integer $l_1>0$ such that the pattern $v$ appears once again in $\mathbf{x}$ with the shift of size $l_1$ to the right,
  $$v = \mathbf{x}_{[l_1:l_1+N-1]}.$$ 
  If $q$ is a factor of $l_1$, then we are done.  Otherwise (if $q$ is not a factor of $l_1$), we use  quasiperiodicity of $\mathbf{x}$ once again (now for a bigger pattern). From quasiperiodicity it follows that there exists an integer $l_2>0$ such that $\mathbf{x}_{[0:l_1+N-1]}$ appears again in $\mathbf{x}$  with the shift of size $l_2$ to the right,
  $$\mathbf{x}_{[0:l_1+N-1]}= \mathbf{x}_{[l_2:l_1+l_2+N-1]}.$$   
Now we have two new occurrences of $v$ in $\mathbf{x}$,
  $$
  v = \mathbf{x}_{[l_1:l_1+N-1]}= \mathbf{x}_{[l_2:l_2+N-1]}= \mathbf{x}_{[l_1+l_2:l_1+l_2+N-1]}.
  $$
If $q$ is a factor of $l_2$ or $l_1+l_2$, we get a subword in $\mathbf{x}$ (starting at the position $l_2$ or $l_1+l_2$ respectively) that is equal to  $v$. 
Otherwise, we repeat the same argument  again, and find in $\mathbf{x}$ a copy of an even greater pattern $\mathbf{x}_{[0:l_1+l_2+N-1]}$. Repeating this argument $k$ times, we obtain  a sequence of positive integers  $l_1,\ldots, l_k$  such that the word $v$ reappears in $\mathbf{x}$ with all shifts composed of terms $l_i$  (all possible sums of several different $l_i$).  That is, for each integer number
\begin{equation}\label{sum-l}
 \sigma = l_{j_1}+l_{j_2}+\ldots + l_{j_r}
\end{equation}
(composed of $r$ pairwise different indices $1\le j_1< \cdots < j_r\le k$) we have 
  $$
  v= \mathbf{x}_{[\sigma:\sigma+N-1]},
  $$
see Fig.~3.
If the number $k$ is large enough, then $q$ is a factor of at least one of the shifts~(\ref{sum-l}). Indeed, if $k>q(q-1)$ then we can find  $q$ different $l_j$ congruent to each other modulo $q$ (the pigeon hole principle). Then the sum of these $l_j$ must be equal to $0$ modulo $q$.  
\medskip

\begin{figure}
\includegraphics[width=\hsize]{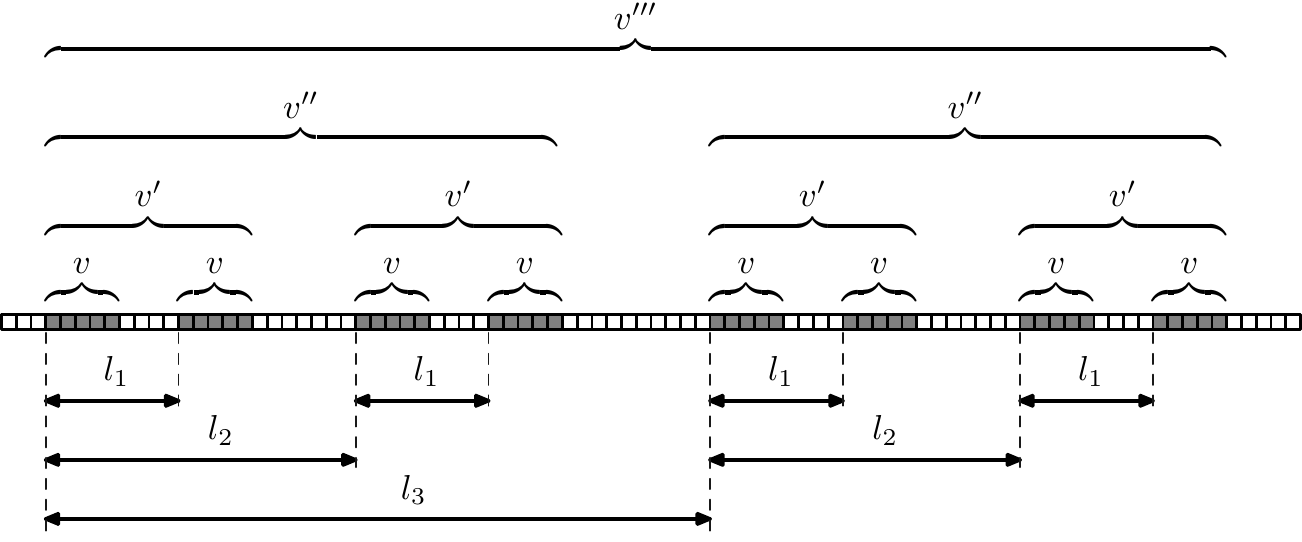}
\caption{Reappearance of factors in a quasiperiodic sequence: at  first we find a reappearance of the factor $v$, then we find a reappearance of the factor $v'$ involving two copies of $v$, then a reappearance of $v''$ involving two copies of $v'$ (and four copies of $v$), etc.}
\end{figure}

\smallskip

So far  we used only the fact that $\mathbf{x}$ is  recurrent. To prove the \emph{moreover} part of the lemma, we notice that for a \emph{uniformly recurrent} $\mathbf{x}$, all integers $l_j$, $j=1,\ldots,k$ in the argument above can be majorized by some  uniform  upper bound that depends only on $\mathbf{x}$ and $N$ (but not on a specific position of the factor $v$ in $\mathbf{x}$ chosen in the first place).  
This observation concludes the proof.

\end{proof}

\begin{proof}[Proof of Lemma~\ref{lemma-minimal-times-periodic}]
(a) Denote by $q$ the period of $\mathbf{y}$.
Since $\cal T$ is  a minimal shift, the configuration $\mathbf{x}$ is recurrent and even quasiperiodic. By Lemma~\ref{lemma-quasiperiodic-times-periodic}, every factor $v=\mathbf{x}_{[m:n]}$ reappears in $\mathbf{x}$ with a shift divisible by $q$. Moreover, a copy of $v$ can be found with a shift divisible by $q$ in every large enough pattern
 $
 \mathbf{x}_{[i: i+L]}
 $
(where the value of $L$ depends on $v$ but not on the specific position of $v$ in $\mathbf{x}$).
It follows that any factor $\tilde v =  (\mathbf{x} \otimes \mathbf{y})_{[m:n]}$ that contains the word $v=\mathbf{x}_{[m:n]}$ in its first component  
reappears in every pattern of size $L=L(v)$ in each sequences of ${\cal S}(\mathbf{x} \otimes \mathbf{y})$.

\smallskip

(b) We need to verify algorithmically  whether a given factor $\tilde v$ appears in $\mathbf{x} \otimes \mathbf{y}$.  In other words, for a given word $v$ and for a given integer $i$ we need to find out whether $v$ appears in $\mathbf{x}$ in a  position congruent to $i\mod q$.

From Lemma~\ref{lemma-quasiperiodic-times-periodic} it follows that 
there exists an $L=L(v)$ such that for every appearance of $v$ in $\mathbf{x}$, 
 this factor reappears in $\mathbf{x}$  in the same position modulo $q$ with a translation at most $L$ to the right. More precisely, 
if $v=\mathbf{x}_{[i:i+|v|-1]}$ for some $i$, than there exists a $\sigma<L$ divisible by $q$ such that $\mathbf{x}_{[i:i+|v|-1]}=\mathbf{x}_{[i+\sigma:i+\sigma+|v|-1]}$.

Since $\cal T$ is minimal, the language of the finite factors of all configurations in $\cal T$  is computable. It follows that the bound $L=L(v)$ defined above is computable. Indeed, by the brute force search we can find the maximal possible gap between two neighboring appearances  (in this subshift) of the word $v$ in  positions congruent to each other modulo $q$.
Thus, to  find out whether $v$ appears in $\mathbf{x}$ in a position congruent $i\mod q$, it is enough to compute the first $L(v)$ letters of the sequence~$\mathbf{x}$.
\end{proof}

\section{Towards quasiperiodic SFT}

In this section we combine the combinatorial lemmas from the previous section with the technique of enforcing quasiperiodicity from \cite{mfcs2015}, and  construct quasiperiodic tilings that simulate quasiperiodic $\mathbb{Z}$-shifts.

\subsection{When macro-tiles are clones of each other}
 
To show that (some) self-simulating tilings enjoy the property of quasiperiodicity, we need a tool to prove that every pattern in a tiling has ``clones'' (equal patterns) in each large enough fragment of this tiling. In our tiling every finite pattern is covered by a block of (at most) four macro-tiles of high enough rank, so we can focus on the search for ``clones'' in macro-tiles. The following lemma gives a natural characterization of the equality 
of two macro-tiles in a tiling: they must have the same information in their  ``conscious memory'' (the data written on the tape of the Turing machine in 
the computation zone) and the same information hidden in their  ``deep subconscious''  (the fragments of the embedded $1$D sequence 
corresponding to the responsibility zones of these macro-tiles must be identical).
  
 \begin{lemma}\label{lemma-clones}
 Two macro-tiles of rank $k$ are equal to each other if and only if they
  (a)~contain the same bits in the fields (i) - (vi) in the input data on the computation zone, and
  (b)~the factors of the encoded sequence corresponding to the zones of responsibility 
  of these macro-tiles \textup(in the corresponding vertical stripes of width $3N_k$\textup) are equal to each other.
 \end{lemma}
 \begin{proof}
 Induction by the rank $k$. For the macro-tile of rank $1$ the statement follows directly from the construction. For a pair of macro-tiles $M_1$
 and $M_2$ of rank  $(k+1)$ with identical data in the fields (i) - (vi) we observe that the corresponding ``cells'' in $M_1$ and $M_2$
 (which are macro-tiles of rank $k$) contain the same data in their own   fields (i) - (vi), since  the communication wires of $M_1$ and $M_2$
 carry the same information bits, their computation zones represent exactly the same 
 computations, etc. If the  factors (of length $3L_k$) from the encoded sequences in the zones of responsibility of  $M_1$ and $M_2$  
 are also equal to each other, we can apply the inductive assumption.
 \end{proof}
 
 \subsection{Supplementary features:  constraints that can be imposed on the self-simulating tiling}

The tiles involved in our self-simulating tiles set (as well as all macro-tile of each rank) can be classified into three types:
 \begin{itemize}
 \item[(a)] the ``skeleton'' tiles that keep no information except for their coordinates in the father macro-tile; these tiles work as building blocks of the hierarchical structure;
 \item[(b)] the ``communication wires'' that transmit the bits of macro-colors from the border line of the macro-tile to the computation zone;
 \item[(c)] the tiles of the computation zone (intended to simulate the space-time diagram of the Universal Turing machine). 
 \end{itemize}
 Each pattern that includes only ``skeleton'' tiles (or ``skeleton'' macro-tiles of some rank $k$) reappears infinitely often in all homologous position inside all macro-tiles of higher rank. Unfortunately, this property is not true for the patterns that involve  the ``communication zone'' or the ``communication wires''.  Thus, the general construction of a fixed-point tiling does not imply the property of quasiperiodicity.
 To overcome this difficulty we need some new technical tricks.

We can enforce the following additional properties (p1) - (p4) of a tiling with only a minor modification of the construction:

\smallskip
\noindent
\textbf{(p1)} In each macro-tile, the size of the computation zone $m_k$ is much less than the size of the macro-tile $N$.  In what follows we  need to reserve free space in a macro-tile to insert $O(1)$ (some constant number) of copies of each $2\times 2$ pattern  from the computation zone (of this macro-tile), right above the computation zone. This requirement is easy to meet.
We assume that the size of a computation zone in a $k$-level macro-tile of size $N_k\times N_k$ is only $m_k=\poly(\log N_k)$.
So we can reserve an area of size $\Omega(m_k)$ above the computation zone, which is free of ``communication wires'' or any other functional gadgets 
(so far this area consisted  of only skeleton tiles),
see the ``empty'' hatched area in Fig.~4.

\smallskip
\noindent
\textbf{(p2)}  We require that the tiling inside the computation zone   satisfies the property of $2\times2$-\emph{determinacy}. If we know all the  colors  on the borderline of a $2\times2$-pattern inside of the computation zone (i.e., a tuple of $8$ colors), then we can uniquely reconstruct the $4$ tiles of this pattern. Again, to implement this property we do not need  new ideas; this requirement is met if we represent the space-time diagram of a Turing machine in a natural way.

 \begin{wrapfigure}{l}{0.45\textwidth}
 \vspace{-5pt}
\includegraphics[scale=0.45]{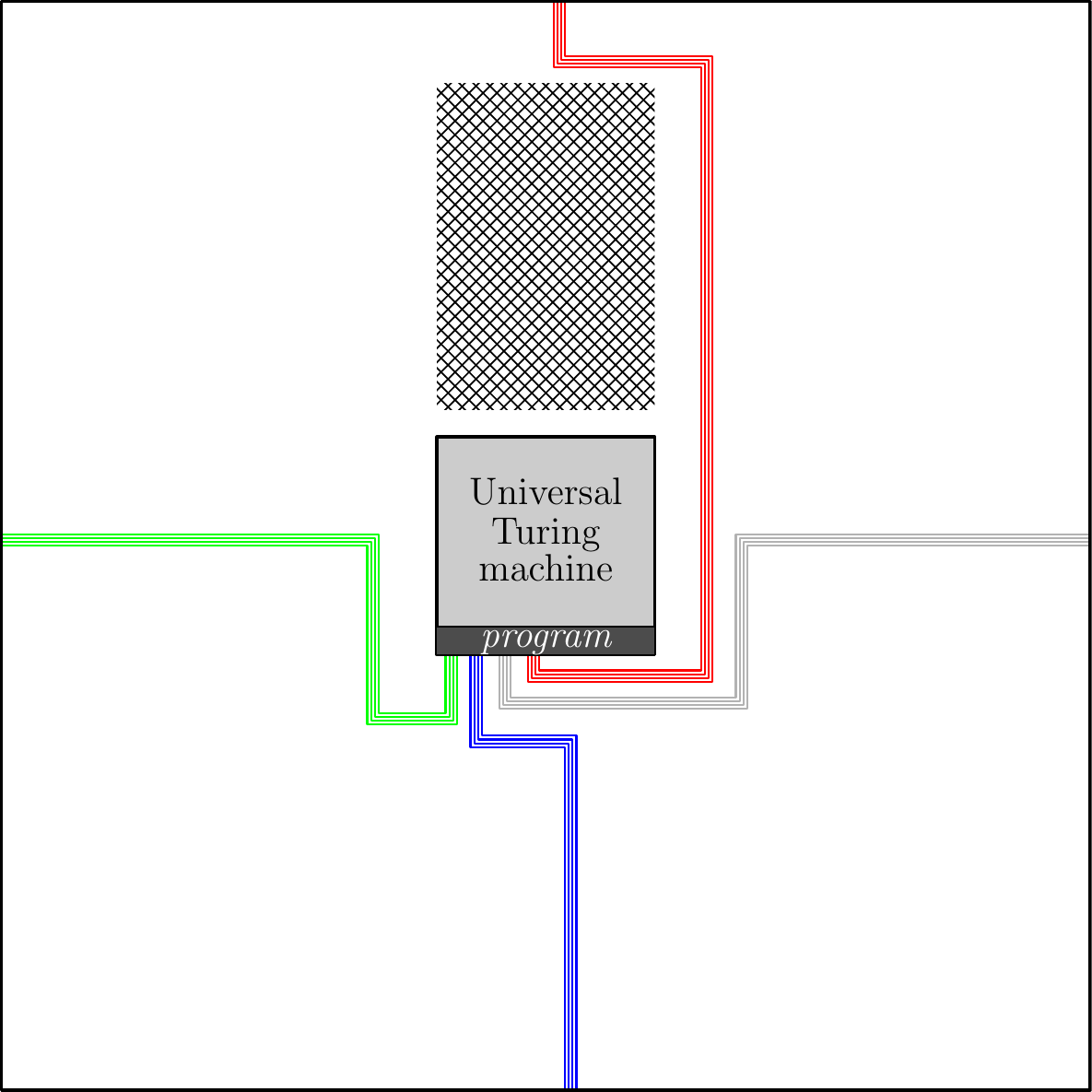}
\label{pic5}
\caption{The ``free'' area reserved above the computational zone}
\vspace{-10pt}
 \end{wrapfigure}

\smallskip
\noindent
\textbf{(p3)}  The communication channels in a macro-tile   (the wires that transmit the information from the macro-color on the borderline of this macro-tile to the bottom line of its computation zone) must be isolated from each other. The distance between every two wires must be greater than $2$. That is, each $2\times 2$-pattern   can touch at most one communication wire. Since the width of the wires in a $k$-level macro-tile is only $O(\log N_{k+1})$, we have enough free space to lay the ``communication cables'', so this requirement is easy to satisfy.

\noindent
\emph{Remark:} Property (p3) is a new feature, it was not used in \cite{mfcs2015} or any other preceding constructions of
self-simulating tilings.

\smallskip
\noindent
\textbf{(p4)}  In our construction the macro-colors of  a $k$-level macro-tile are encoded by bit strings of some length $r_k = O(\log N_{k+1})$.
We  assumed that this encoding is natural in some way. So far the choice of encoding was of small importance; we only required that some natural manipulations with macro-colors can be implemented in polynomial time.
Now, we add another (seemingly artificial) requirement: that each of $r_k$ bits encoding the macro-colors (on the top, bottom, left and right sides of a macro-color) was equal to $0$ and to $1$ for quite a lot of macro-tiles (so the fact that some bit of some macro-color has this or that value, must not be unique in a tiling). Technically, we require an even stronger property: at every position $s=1,\ldots, r_k$ and for every $i=0,\ldots,N_{k+1}-1$ there must exist  $j_0, j_1$ such that the $s$-th bit in the top, the left and the right macro-colors of the $k$-level macro-tile at the positions $(i,j_0)$ and $(i,j_1)$ in the $(k+1)$-level father macro-tile  is equal to  $0$ and $1$ respectively.

There are many (more or less artificial) ways to realize this constraint. For example, we may subdivide the array of $r_k$ bits in three equal zones of size $r_k/3$ and require that for each macro-tile only one of these three zones contains the ``meaningful'' bits, and two other zones contain only zeros and ones respectively; we require then that the ``roles'' of these three zones cyclically exchange as we go upwards along a column of macro-tiles.

\subsection{Enforcing quasiperiodicity}

To achieve the property of quasiperiodicity, we should guarantee  that every finite pattern that appears once in a tiling, must appear in each large enough square.  
If a tileset $\tau$ is self-similar, then in every $\tau$-tiling  each finite pattern  can be covered by at most $4$  macro-tiles  (by a $2\times2$-pattern)  of an appropriate rank. Thus, it is enough to show that every $2\times 2$-block of macro-tiles of any rank $k$ that appears in at least one $\tau$-tiling, actually appears in this tiling in every large enough square.

\smallskip
\emph{Case 1: skeleton tiles.}
For a $2\times 2$-block of four ``skeleton'' macro-tiles of level $k$ this is easy. Indeed, we have exactly the same blocks with  every vertical shift multiple of  $L_{k+1}$ (we have there a similar block of $k$-level ``skeleton'' macro-tiles within another  macro-tile of rank $(k+1)$). A vertical shift does not change the embedded letters in the zone of responsibility, so we can apply Lemma~\ref{lemma-clones}. 

To find a similar block of $k$-level ``skeleton'' macro-tiles with a different abscissa coordinate, we need a horizontal shift $Q$ which is divisible by $L_{k+1}$  (to preserve the position in the father macro-tile) and at the same time does not change the letters embedded in the zone of responsibility. This is possible due to Lemma~\ref{lemma-quasiperiodic-times-periodic}, if the embedded sequence is quasiperiodic. Given a suitable horizontal shift,  we can again apply Lemma~\ref{lemma-clones}.

\smallskip
\emph{Case 2: communication wires.}
Let us consider the case when a $2\times 2$-block of $k$-level macro-tiles involves a part of a communication wire. Due to the property (p3) we may assume that only one wire is involved. The bit transmitted by this wire is either $0$ or $1$; in both cases, due to the property (p4)
we can find another similar $2\times 2$-block of $k$-level macro-tiles (at the same position within the father macro-tile of rank $(k+1)$ 
and with the same bit included in the communication wire)  in  every macro-tile of level $(k+2)$. 
In this case we need a vertical shift  longer than in Case~1: we can find a duplicate of the given block with a vertical shift of size $O(L_{k+2})$.  

As in  Case~1, any vertical shift does not change  the letters embedded in the zone of responsibility of the involved macro-tiles, and we can apply  Lemma~\ref{lemma-clones} immediately.  If we are looking for a horizontal shift, we again use quasiperiodicity of the simulated shift and apply Lemma~\ref{lemma-quasiperiodic-times-periodic}: there exists a horizontal shift that is divisible by $L_{k+2}$ and  does not change the letters embedded in the zone of responsibility. Then we again apply Lemma~\ref{lemma-clones}.

\smallskip
\emph{Case 3: computation zone.} Now we consider the most difficult case: when  a $2\times 2$-block of $k$-level macro-tiles touches the computation zone.
In this case we cannot obtain the property of quasiperiodicity for free, and we have to make one more (the last one) modification  of our general construction of a self-simulating tiling.

\begin{figure}
\includegraphics[scale=0.5]{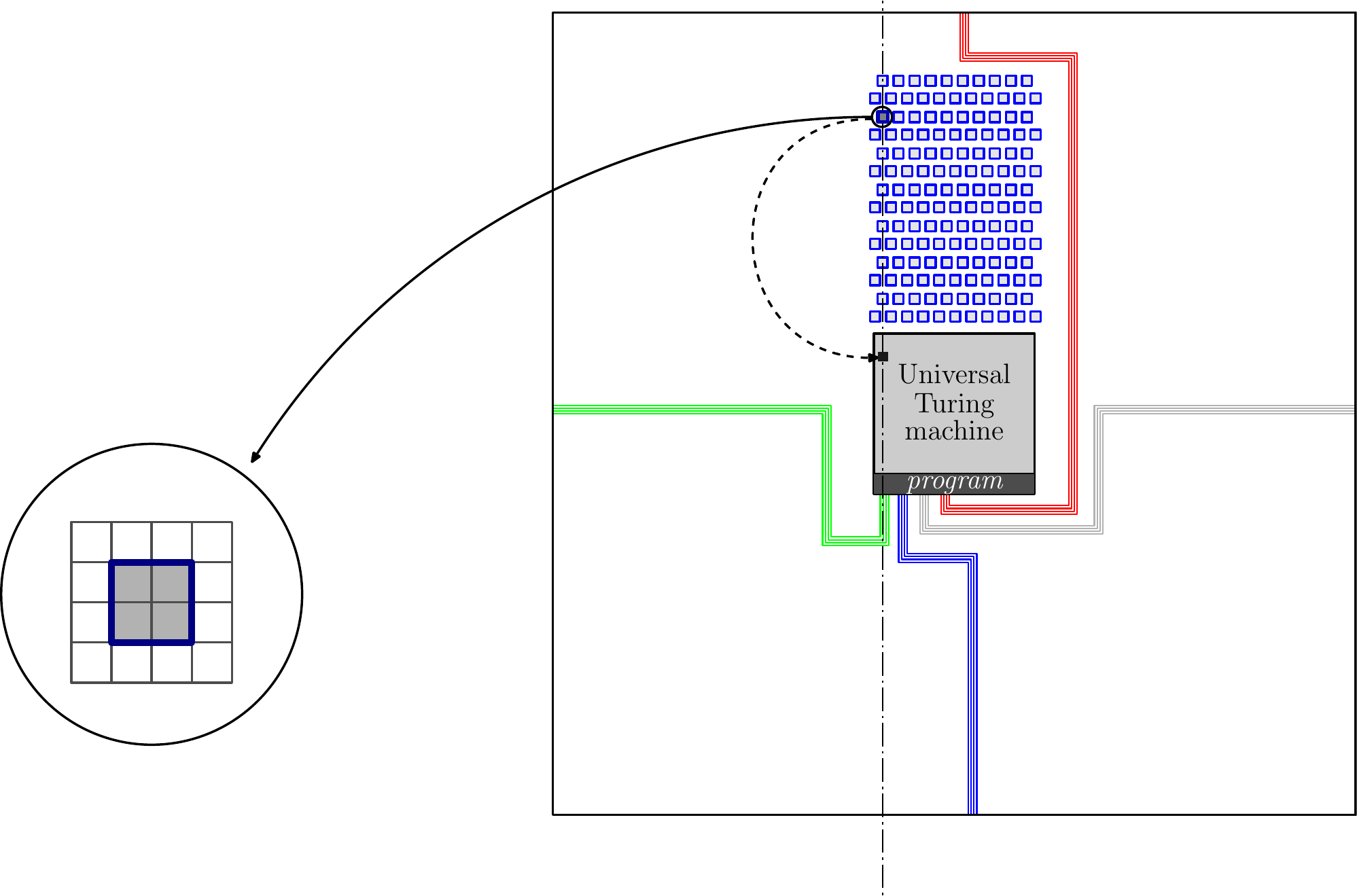}
\label{pic6}
\caption{Positions of the slots for patterns from the computation zone.}
\end{figure}

 Notice that for each $2\times2$-window that touches the computation zone  of a macro-tile there exist only $O(1)$ ways to tile them correctly. For each possible position of a $2\times2$-window in the computation zone and for each possible filling of this window by tiles, we reserve a special $2\times2$-\emph{slot} in a macro-tile, which is essentially a block of size $2\times2$ in the ``free'' zone of a macro-tile. It must be placed   far away from the computation zone and from all communication wires, but in the same vertical stripe as the ``original'' position of this block, see Fig.~5.
 We have enough free space to place all necessary  slots due to the property (p1). We  define the neighbors around this slot in such a way that only one specific $2\times 2$ pattern  can patch it (here we use the property (p2)).

 \begin{wrapfigure}{r}{0.5\textwidth}
\centering
\vspace{-5pt}
\includegraphics[scale=0.33]{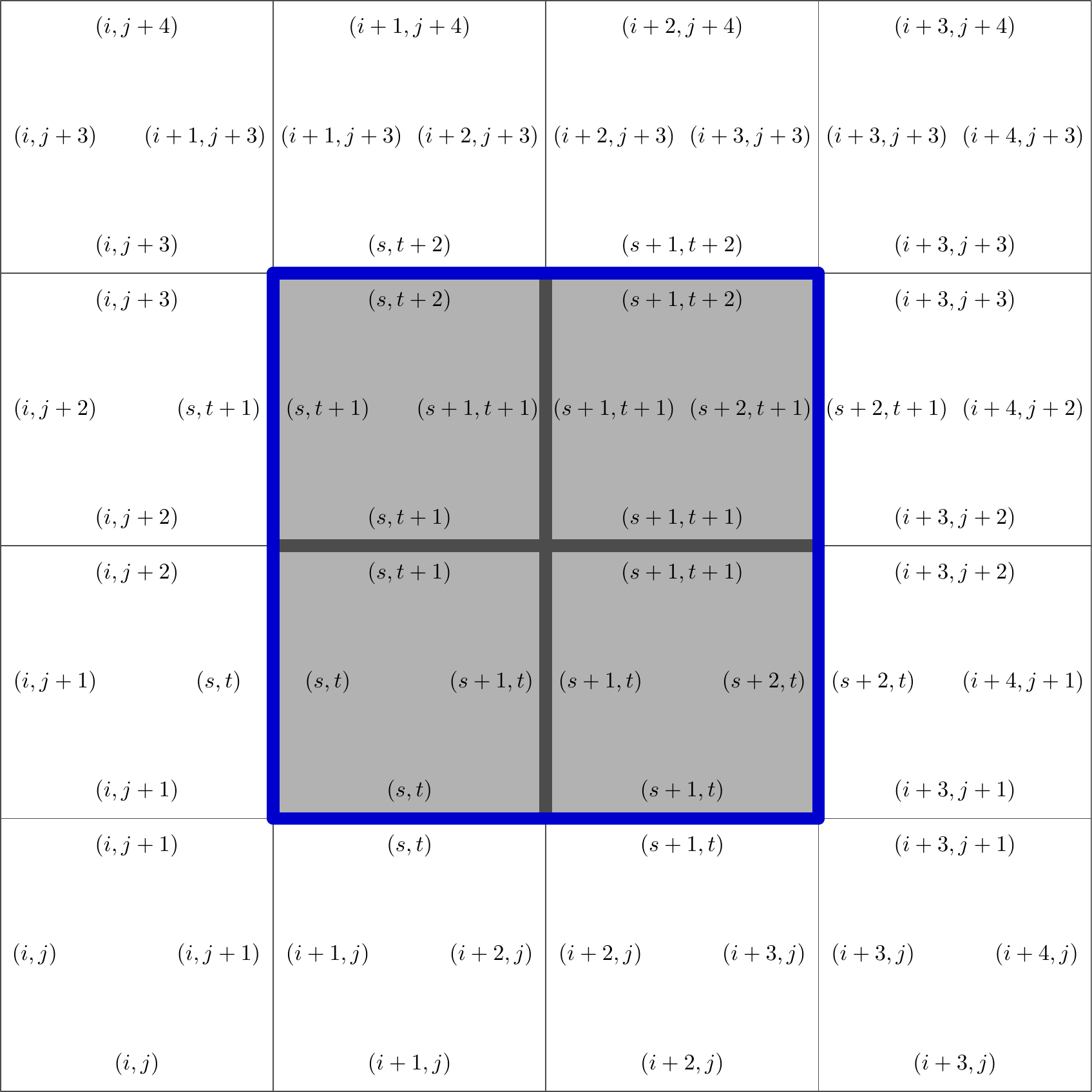}
\caption{A slot for a {$2\times2$}-pattern from the computation zone.}
\vspace{-5pt}
\end{wrapfigure}

In our construction the tiles around this  slot ``know'' their real coordinates in the bigger macro-tile, while the tiles inside the slot do not (they ``believe'' they are tiles in the computation zone, while in fact they belong  to an artificial isolated  \emph{diversity preserving} ``slot'' far outside of any real computation), see Fig.~5 and Fig.~6.
 The frame of the slot consists of 12 ``skeleton'' tiles (the white squares in Fig.~6), they form a slot a $2\times 2$-pattern  from the computation zone (the grey squares in Fig.~6). 
In the picture we show the ``coordinates'' encoded in the colors on the sides of each tile.  
Notice that the colors of the bold lines (the blue lines between white and grey tiles and the bold black lines between grey tiles) should contain some information beyond coordinates --- these colors involve the bits used to simulate a space-time diagram of the universal Turing machine. 
(We do not show all the corresponding bits explicitly.) 
In this picture, the ``real'' coordinates of the bottom-left corner of this slot are $(i+1,j+1)$, 
while the ``natural'' coordinates of the  pattern (when it appears in the computation zone) are $(s,t)$. 

We choose the positions of the ``slots'' in the macro-tile so that coordinates can be computed with a short program in time polynomial in $\log N$. We require that all slots are isolated from each other in space, so they do not damage the general structure of ``skeleton'' tiles building the macro-tiles.

Through construction, each of these slots is aligned with the ``natural'' position of the corresponding $2\times 2$-block in the computation zone. This guarantees that the tiles in the computation zone and their ``sibling'' in the artificial slots  share the same bits of the embedded sequences in the corresponding zone of responsibility. We have defined the slots so that the ``conscious memory'' of the tiles in the computation zone and in the corresponding slots is the same. 
Thus, we can apply Lemma~\ref{lemma-clones} and conclude that a  $2\times 2$-blocks in diversity preserving slots are exactly equal to the corresponding $2\times 2$-patterns in the computation zone.  

For a horizontal shift, similarly to the Cases~1--2 above, we use quasiperiodicity of the embedded sequences and apply Lemma~\ref{lemma-quasiperiodic-times-periodic}.

\smallskip
\emph{Concluding remark:} Formally speaking, we  proved Lemma~\ref{lemma-clones} before we introduced the last upgrades of our tileset. However, it is easy to verify that the updates  of the main construction discussed in this Section do not affect the proof of that lemma.
\smallskip

Thus, we constructed a tileset $\tau$ such that every $L_k\times L_k$ pattern that appears in a $\tau$-tiling must also appear in every large enough square in this tiling. So, the constructed tileset satisfies the requirements of Theorem~\ref{thm-main}.

\begin{proof}[The proof of Corollary~\ref{thm-kolmogorov}.]
To prove Corollary~\ref{thm-kolmogorov} we only need to combine Theorem~\ref{thm-main} with a fact from  \cite{rumyantsev-ushakov}:  there exists a $1$D shift $\cal S$ that is quasiperiodic, and   for every configuration $\mathbf{x}\in{\cal S}$ the Kolmogorov complexity of all factors is linear, i.e.,
 $
 K(x_{i}x_{i+1}\ldots x_{i+n})  = \Omega(n)
 $
for all $i$.
\end{proof}

%
\noindent\textbf{The proof of Theorem~\ref{thm-main-min}.}
First of all we notice that the proof of Theorem~\ref{thm-main} discussed above does not imply Theorem~\ref{thm-main-min}. 
If we take an effective 
 minimal $1\mathrm{D}$-shift $\cal A$ and plug it into the construction form the proof of  Theorem~\ref{thm-main}, we obtain a tileset  $\tau$
(simulating $\cal A$) which is \emph{quasiperiodic} but not necessary \emph{minimal}. The property of minimality can be lost even for a \emph{periodic} shift $\cal A$. Indeed, assume that the minimal period $t>0$ of the configurations in $\cal A$ is a factor of the size $N_k$ of $k$-level macro-tiles in our self-simulating tiling, then  we can extract from the resulting  SFT $\tau$ nontrivial shifts $T_i$, $i=0,1,\ldots, t-1$ corresponding to the position of the embedded 
$1\mathrm{D}$-configuration with respect to the grid of macro-tiles.
To overcome this obstacle we will superimpose some additional constraints on the embedding of the simulated $\mathbb{Z}$-shift in a $\mathbb{Z}^2$-tiling. Roughly speaking, we will enforce only ``standard'' positioning of the embedded  $1\mathrm{D}$ sequences with respect to the grid of macro-tiles. This will not change the class of configurations that can be simulated (we still get all configurations from a given minimal shift $\cal A$), but the classes of all valid tilings will reduce to some minimal $\mathbb{Z}^2$-SFT.

\emph{The standardly aligned grid of macro-tiles:} In general, the hierarchical structure of macro-tiles permits non-countably many ways of cutting the plane in macro-tiles of different ranks. We fix one particular version of this hierarchical structure and say that a grid of macro-tiles is \emph{standardly aligned}, if for each level $k$ the point $(0,0)$ is the bottom-left corner of a $k$-level macro-tile
(see Fig.~7). 
 This means that the tiling is cut into $k$-level macro-tiles of size $L_k\times L_k$ by vertical lines with abscissae $x=L_k\cdot t'$ and ordinates  $y=L_k\cdot t''$, with $t', t''\in \mathbb{Z}$
(so the vertical line $(0,*)$ and the horizontal line $(*,0)$ serve as separating lines for macro-tiles of all ranks). 
Of course, this structure of macro-tiles is computable.

\begin{wrapfigure}{L}{0.5\textwidth}
\centering
\vspace{0pt}
\includegraphics[scale=0.42]{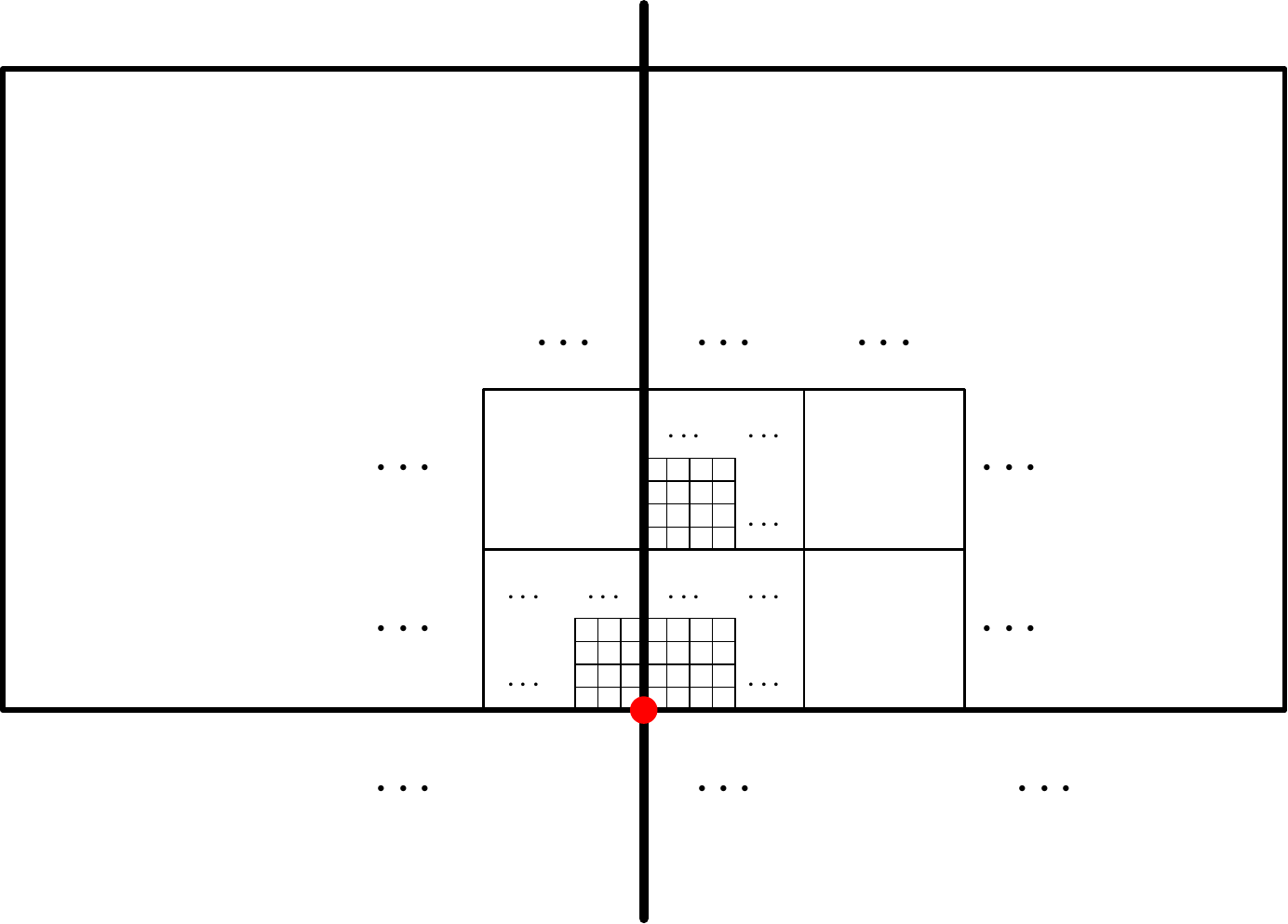}
\caption{The standardly aligned grid of macro-tiles: the central point (marked red) is a corner of macro-tiles of each rank.}
\vspace{-5pt}
\end{wrapfigure}

\smallskip

\emph{The canonical representative of a minimal shift:} A minimal effectively-closed  $1\mathrm{D}$-shift $\cal A$ is always computable, i.e., the set of finite patterns that appear in configurations of this shift is computable. It follows immediately that $\cal A$ contains some computable configuration. Let us fix one computable configuration $\mathbf{x}$; in what follows we call it  \emph{canonical}.

\smallskip

\emph{The standard  embedding of the canonical representative:} We  superimpose the standardly aligned grid of macro-tiles with the canonical representative of a minimal shift $\cal A$:  we take the direct product of the hierarchical structures  of the standardly aligned  grid of macro-tiles
with  the {canonical} configuration $\mathbf{x}$ from ${\cal A}$  (that is, each tile with coordinates $(i,j)$ ``contains'' the letter $x_i$ from the {canonical} configuration).

\smallskip

\emph{Claim 1:} Given a pattern $w$ of size $n\le L_k$ and an integer $i$, we can algorithmically verify whether the factor $w$ appears in the standard embedding of the canonical representative with the shift $(i \mod L_k)$ relative to the grid of $k$-level macro-tiles. This follows from Lemma~\ref{lemma-minimal-times-periodic}(b) applied to the superposition of the canonical representative with the  
periodical grid of $k$-level macro-tiles).

\emph{Remark:} This verification procedure is computable, but its computational complexity can be very high. 
To perform the necessary  computation we may need space and time much bigger than the length of $w$ and $L_k$.

\smallskip

\emph{Upgrade of the main construction:} Let us update the construction of self-simulating tiling from the proof of Theorem~\ref{thm-main}. So far we assumed that every macro-tile (of every level $k$) verifies that the delegated factor of the embedded sequences contains no factors forbidden for the shift $\cal A$. Now we make the constraint stronger: we require that the delegated factor contains only factors allowed in the shift $\cal A$ and placed in the positions (relative to the grid of macro-tiles) permitted for factors in the standard  embedding of the canonical representative. This property is computable (Claim~1), so every  forbidden pattern or a 
pattern in a forbidden position will be discovered in a computation in a macro-tile of some rank.

The computational complexity of this procedure can be very high (see Remark after Claim~1), and we cannot guarantee that the forbidden patterns of small length are discovered by the computation in macro-tiles of small size. But we do guarantee that each forbidden pattern or a pattern in a forbidden position is discovered by a computation in some macro-tile of\emph{ high enough} rank.

\smallskip

\emph{Claim 2:} The new tileset admits correct tilings of the plane. Indeed, at least one tiling is valid by the construction: the standard  embedding of the canonical representative corresponds to a valid tiling of the plane,  since macro-tiles of all rank never find any forbidden placement of patterns in the embedded sequence.

\smallskip

\emph{Claim 3:} The new tileset simulates the shift $\cal A$. This follows immediately from the construction: the embedded sequence must be a configuration from $\cal A$.

\smallskip

\emph{Claim 4:} For the constructed tileset $\tau$ the set of all tilings is a minimal shift. We need to show that \emph{every} $\tau$-tiling contains all patterns that can appear in \emph{at least one} $\tau$-tiling. Similarly to the proof of  Theorem~\ref{thm-main}, it is enough to prove this property for $2\times2$-blocks of 
$k$-level macro-tile. The difference with the argument in the previous section is that for every $2\times2$-block  of macro-tiles in one tiling $T$ we must find a similar block of macro-tiles in another tiling $T'$, so that this block has exactly the same position with respect to father macro-tile $\cal M$ of rank $(k+1)$,  and $\cal M$ and $\cal M'$  own exactly the same factor of the embedded sequence in their zones of responsibility. This is always possible due to Lemma~\ref{lemma-minimal-times-periodic}(a) (applied to the canonical representative of $\cal A$ superimposed with the periodical grid of $(k+1)$-level macro-tiles). This observation concludes the proof.

\medskip
\noindent
\textbf{Acknowledgments.} 
 We are 
 indebted to Emmanuel Jeandel for raising  and motivating the questions
  which led to this work.
We  are grateful to  Gwena\"el Richommes and Pascal Vanier  
 for  fruitful discussions and 
 to the anonymous reviewers of the MFCS-2017 for truly valuable  comments.


\bibliography{arxiv_revised}

\appendix

\section{A more detailed explanation of the fixed-point tiling}

In this appendix we explain in more detail the construction of a ``self-simulating'' tileset from \cite{drs}. Basically, this appendix is 
an elaborate version of the preamble of Sections~2.

\subsection{The relation of simulation for  tilesets}

Let $\tau$ be a tileset and $N>1$ be an integer. We call a \emph{macro-tile} an $N \times N $ square  tiled by matching tiles from $\tau$. Every side of a $\tau$-macro-tile contains  a sequence of $N$ colors (of tiles from $\tau$); we refer to this sequence as a \emph{macro-color}.
Further, let $T$ be some set of $\tau$-macro-tiles (of size $N\times N$). We say that $\tau$ \emph{implements} $T$ with a  \emph{zoom factor} $N$,
if (i) some $\tau$-tilings exist, and (ii) for every $\tau$-tiling there exists a unique lattice of vertical and horizontal lines that cuts this tiling into $N\times N$ macro-tiles from $\rho$.  A tileset $\tau$ \emph{simulates} another tileset $\rho$, if $\tau$ implements a set of macro-tiles $T$ (with a zoom factor $N>1$) that is isomorphic to $\rho$, i.e.,  there exists a one-to-one correspondence between $\rho$ and $T$ such that the matching pairs of $\rho$-tiles correspond exactly to the matching pairs of $T$-macro-tiles.  A  tileset $\tau$ is called \emph{self-similar} if it simulates itself.

If a  tileset $\tau$ is self-similar, then all $\tau$-tilings have a hierarchical structure. Indeed,  each $\tau$-tiling can be uniquely split into $N\times N$ macro-tiles from a set $T$,  and these macro-tiles are  isomorphic to the initial tileset $\tau$.  Further, the grid of macro-tiles can be uniquely grouped into blocks of size $N^2\times N^2$, where each block is a macro-tile of rank $2$ (again, the set of all macro-tiles of rank $2$ is isomorphic to the initial tileset $\tau$), etc. It is not hard to deduce  that a self-similar tileset $\tau$ has only aperiodic tilings (for more detail see  \cite{drs}).  Below, we discuss a generic construction of self-similar tilesets.

\subsection{Simulating a tileset defined by a Turing machine}

Let us have a tileset $\rho$ where each color is a $k$-bit string (i.e., the set of colors $C \subset \{0,1\}^k$) and the set of tiles $\rho \subset C^4$ is presented by a predicate $P(c_1,c_2,c_3,c_4)$ (the predicate is true if and only if the quadruple $(c_1,c_2,c_3,c_4)$ corresponds to a tile from $\rho$). Let us have some Turing machine $\cal M$ that computes $P$. In what follows we present a general construction that allows us to simulate  $\rho$ by some other tileset $\tau$, with a large enough zoom factor $N$.

 \begin{figure}
\center
\center
\includegraphics[scale=0.50]{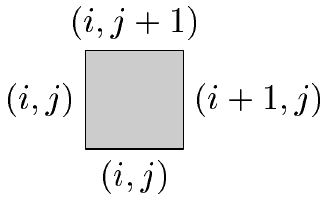}%
\rule{2.0cm}{0.0cm}%
\includegraphics[scale=0.8]{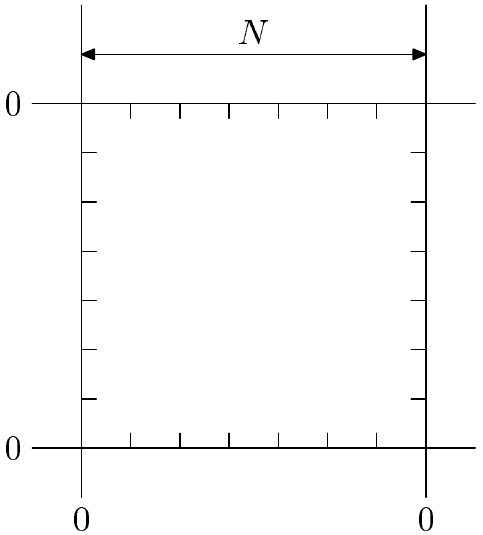}
\caption{A tile with coordinates on the sides and the macro-tile grid.}
\end{figure}

We  build a tileset $\tau$ where each tile ``knows'' its coordinates modulo $N$ in the tiling. This information is included in the tile colors. More precisely, for a tile that is supposed to have  coordinates $(i,j)$ modulo $N$, the colors on the left and on the bottom sides should involve $(i,j)$, the color on the right side should involve $(i+1\mod N, j)$, and the color on the top side, respectively, involves $(i, j+1 \mod N)$, see Fig.~8.
This means that every $\tau$-tiling can be uniquely split into blocks (macro-tiles) of size $N\times N$, where the coordinates of cells range from $(0,0)$ in the bottom-left corner to $(N-1,N-1)$ in top-right corner, Fig.~8. Intuitively, each tile ``knows'' its position in the corresponding macro-tile.

 We require that in addition to the coordinates, each tile in $\tau$  has some supplementary information   encoded in the colors on its sides.  On the border of a macro-tile (where one of the coordinates is zero) we assign to the colors of tiles one additional bit of information. Thus, for each macro-tile of size $N\times N$ the corresponding macro-colors can be  represented as strings of $N$ zeros and ones.  Further, we assume that $k\ll N$.  We allocate  $k$ positions in the middle of a macro-tile sides and make them represent colors from $C$; all other bits on the sides of a macro-tile are zeros.
  
 Now we introduce additional restrictions on tiles in $\tau$ that will guarantee that the macro-colors on the macro-tiles satisfy the ``simulated'' relation $P$. To this end we ensure that bits from the macro-tile side are transferred to the central part of the tile, and the central part of a macro-tile is used to simulate  a  computation of the predicate $P$.
 
We fix which cells in a macro-tile are ``communication wires'' 
and then require that these tiles carry the same (transferred) bit on two sides. The central part of a macro-tile (of size, say $m\times m$, where $m\ll N$) should represent a time-space diagram of the machine $\cal M$ (the tape is horizontal, and time goes up). This is done in a standard way. We require that computation terminates in an accepting state (if not, no correct tiling can be formed).

  To make this construction work, the size of a macro-tile (the integer $N$) should be large enough. First, we need enough room to organize the ``communication wires'' that transfer the bits of macro-colors to the ``computational zone''; second, we need enough time  and space  in the computational zone so that all accepting computations of $\cal M$ terminate in time $m$ and on space $m$.
 
In this construction the number of additional  bits encoded in colors of tiles
depends on the choice of the machine $\cal M$. To avoid this dependency, we replace $\cal M$ by a fixed universal Turing machine $\cal U$ that runs a program simulating $\cal M$.  We may assume that the tape has an additional read-only layer. Each cell of this layer carries a bit that never  changes during the computation; these bits are used as a program for the universal machine. So in the computation zone the columns carry unchanged bits; the construction of a tileset  guarantees that these bits form the program for $\cal U$, and the computation zone of a macro-tile represents a view of an accepting computation for that program, see Fig.~1. In this way we get a tileset $\tau$ that has $O(N^2)$ tiles and simulates $\rho$. This construction works for all large enough $N$.

Of course, the tileset $\tau$ depends on the program simulated in the computational zone. However, this dependency is very limited. The simulated program (and, implicitly, the predicate $P$) affects only the rules for the tiles used in the bottom line of the computational zone. The colors on the sides of all other tiles are universal and do not depend on the simulated tileset~$\rho$.

\subsection{Self-simulation with Kleene's recursion trick}

We have explained  how to implement a given tileset $\rho$ by another tileset $\tau$ with large enough zoom factor $N$. Now we want $\tau$ be isomorphic to $\rho$. This can be done using a trick similar to the proof of Kleene's recursion theorem.
We have noticed above that most steps of the construction of  $\tau$ do not depend on the program for $\cal M$. Let us fix these rules as a part of $\rho$'s definition and set $k = 2 \log N + O(1)$, so that we can encode $O(N^2)$ colors by $k$ bits.  From this definition we  obtain a program $\pi$ that takes $N$ as an input and  that checks that macro-tiles behave like $\tau$-tiles in this respect. We are almost done with the program $\pi$. The only remaining part of the rules for $\tau$ is the hardwired program. We need to guarantee that the computation zone in each macro-tile carries the very same program $\pi$. But since the program (the list of instructions interpreted by the universal Turing machine) is written on the tape of the universal machine, this program can be instructed to access the bits of its own text  and check that if a macro-tile belongs to the computation zone, this macro-tile carries the correct bit of the program.

It remains to choose the parameters $N$ and $m$. We need them to be large enough so the computation described above   (which deals with inputs of size $O(\log N)$) can fit in the computation zone. The computations are rather simple (polynomial in the input size, i.e., polynomial in $O(\log N))$, so they certainly fit in space and time bounded by $m=\poly(\log N)$. This completes the construction of a self-similar aperiodic tileset.
Now, it is not hard to verify that the constructed tilesets (i) allow a tiling of the plane, and (ii) each tiling is self-similar.

\subsection{A more flexible construction: the choice of the zoom factor}

The construction described above works well for all large enough zoom factors $N$. In other words, for all large enough $N$ we get a self-similar tileset $\tau_N$, and the tilings for all these $\tau_N$ have very similar structure, with macro-tiles as shown in Fig.~2. We assume that the position of the ``computational zone'' and the ``communication wires'' in a macro-tile are defined by some simple natural rules, so the ``geometry'' of  macro-tiles for $\tau_N$ can be easily computed given $N$.

So far we have assumed that $N$ is hardwired in the Turing machine $\pi_N$ that computes the predicate defining our self-similar tileset $\tau_N$. Technically, the machine $\pi_N$ takes a tuple of $4$ strings (the strings of bits of length $k=k(N)$ representing the $4$ macro-colors of a macor-tile) as the input  and checks whether these strings represent four colors of one tile in our self-similar tileset $\tau_N$.

Notice that while all $\pi_N$ do substantially the same computations, their dependence of $N$ is quite limited.
It seems instructive to separate the parameter $N$ from the program, so we can slightly update our construction. Instead of many individual programs $\pi_N$ (one program for each zoom factor) we take one generic program $\pi$ that gets the binary expansion of $N$ as an input, and then computes a self-similar tileset with macro-tiles of size $N\times N$. More precisely, $\pi$ takes five inputs: a binary expansion of $N$ and a four strings of $k=k(N)$ bits representing the macro-colors of a macro-tile. Now we require that the binary expansion of $N$ is written on the tape of the simulated Turing machine (in the bottom line of the computational zone, as well as the text of the program $\pi$).  The program can access this values of $N$ while performing the computation. Among other things, the program $\pi$ guarantees that the next level macro-tiles contain the same text of the program and the same parameter $N$ in their computational zones.

\subsection{Generalized self-simulation: variable zoom factor}

The construction of a self-similar tiling can be easily generalized to obtain a \emph{variable zoom factor}. 
 This means that macro-tiles of different ranks are not literally isomorphic to the tiles of the ground level.  We can organize the self-simulation so
that the size of the macro-tiles of rank $k$ in the hierarchy is equal to $N_k\times N_k$, for some suitable sequence of zooms $N_k$, $k=1,2,\ldots$ 
In our principal construction we assume that the zoom factor grows rather fast,  $N_k = C^{3^k}$ for some constant $C$.   To implement this construction, we need that each macro-tile of rank $k$ ``knows'' its own rank. Technically, we require now that the binary expansion of the rank $k$ is written on the tape of the Turing machine simulated on the computation zone. This data is used by a macro-tile to simulate the next level macro-tiles properly.  The size of the computational zone $m_k$ should also grow as a  function of rank $k$ (and be easily computable from $k$); as before, we assume that $m_k = \poly(\log N_k)$.

Thus, in the updated construction the first line of the computational zone contains the following \emph{fields of the input data}:
 \begin{itemize}
  \item[(i)] the program of the simulated Turing machine
 \item[(ii)] the binary expension of the rank of this macro-tile (the level of the hierarchy), $k=1,2,\ldots$,
 \item[(iii)] the bits encoding the macro-colors,  which consist of  the coordinates in its ``father'' macro-tile of rank $(k+1)$ (modulo $N_{k+1}$) plus some
 $O(1)$ supplementary bits of information.
 \end{itemize}
 Notice that we do not need to provide the value of $N_k$ explicitly as an input of  the computation,  since it can be computed from the rank $k$. 
For more detail see \cite{drs}. 

Besides the principal computation required for the sake of self-simulation, we can embed in the computational zones of macro-tiles some supplementary ``payload'' --- some ``useful'' computation that has nothing to do with self-simulation. Since the zoom factor grows with the rank, on each next level we can allocate more and more space and time to this secondary computation process.

\end{document}